\documentclass[reqno,12pt]{amsart}

\usepackage[margin=1in]{geometry}
\usepackage{xspace}
\usepackage[utf8]{inputenc}
\usepackage{graphicx}
\usepackage{amssymb}
\usepackage{mathrsfs}
\usepackage[active]{srcltx}
\usepackage{url}
\usepackage{hyperref}
\usepackage{amsmath,amstext,amsxtra,amsgen,amsbsy,amsopn,amscd,amsthm,amsfonts}
\usepackage{latexsym}
\usepackage{dsfont}
\usepackage{enumitem}

\newtheorem{proposition}{Proposition}
\newtheorem{lemma}[proposition]{Lemma}

\theoremstyle{definition}

\theoremstyle{remark}

\newtheorem{remark}[proposition]{Remark}

\newcommand\R{{\ensuremath {\mathbb R} }}
\newcommand\C{{\ensuremath {\mathbb C} }}

\newcommand\Z{{\ensuremath {\mathbb Z} }}

\renewcommand\phi{\varphi}

\renewcommand\le{\leqslant}
\renewcommand\ge{\geqslant}
\renewcommand\epsilon{\varepsilon}

\renewcommand\tilde{\widetilde}
\renewcommand\bar{\overline}

\newcommand{\gS}{\mathfrak{S}}

\newcommand\ii{{\ensuremath {\infty}}}

\newcommand{\norm}[1]{ \left| \! \left| #1 \right| \! \right| }

\newcommand{\cF}{\mathcal{F}}

\newcommand\1{{\ensuremath {\mathds 1} }}

\newcommand{\prob}{\mathbb{P}}
\newcommand{\esp}{\mathbb{E}}

\newcommand\cS{\mathcal{S}}

\DeclareMathOperator{\tr}{Tr}

\date{}

\title[Littlewood-Paley and Lieb-Thirring]{Littlewood-Paley decomposition of operator densities \\ and application to a new proof of the Lieb-Thirring inequality}

\author{Julien Sabin}

\address{D\'epartement de Math\'ematiques (UMR 8628), Facult\'e des Sciences d'Orsay, Universit\'e Paris-Sud, F-91405 Orsay Cedex}

\email{Julien.Sabin@math.u-psud.fr}

\begin{document}

\begin{abstract}
The goal of this note is to prove a analogue of the Littewood-Paley decomposition for densities of operators and to use it in the context of Lieb-Thirring inequalities.
\end{abstract}

\maketitle

\section*{Introduction}

Let $d\ge1$ and $\psi$ a smooth function on $\R^d$, supported in $\R^d\setminus\{0\}$, satisfying
\begin{equation}\label{eq:partition-unity}
  1=\sum_{j\in\Z}\psi(2^{-j}\xi),\quad\forall\xi\in\R^d\setminus\{0\}.
\end{equation}
An example of such a function is given in \cite[Lemma 8.1]{MusSch-book}. In particular, the function $\psi$ can be chosen to be radial and non-negative. We define the Littlewood-Paley multiplier localizing on frequencies $|\xi|\sim 2^j$ by 
$$P_ju:=\cF^{-1}\left(\xi\mapsto\psi_j(\xi)\cF u(\xi)\right),\ \psi_j:=\psi(2^{-j}\cdot),\ j\in\Z, u\in\cS'(\R^d),$$
where $\cF$ denotes the Fourier transform. The Littlewood-Paley theorem \cite[Thm. 8.3]{MusSch-book} states that for any $1<p<\ii$, there exists $C>0$ such that for any $u\in L^p(\R^d)$ one has 
\begin{equation}\label{eq:LP}
  \frac1C\norm{u}_{L^p} \le \norm{\left(\sum_{j\in\Z} |P_ju|^2\right)^{1/2}}_{L^p} \le C\norm{u}_{L^p}.
\end{equation}
This harmonic analysis result has countless applications, from functional inequalities to nonlinear PDEs. It allows to obtain information about $L^p$-properties of a function $u$ from the frequency-localized pieces $P_ju$. For instance, it leads to a very short proof of the Sobolev embedding $H^s(\R^d)\hookrightarrow L^p(\R^d)$ for $p=2d/(d-2s)$, $0<s<d/2$, as we recall in Section \ref{sec:comparison}. It was also used, for instance, to prove Strichartz-type inequalities \cite{KeeTao-98,BurGerTzv-04}. We refer to \cite{FraJawWei-book} for more general applications of Littlewood-Paley theory. 

This note is devoted to a generalization of \eqref{eq:LP} to densities of operators. When $\gamma\ge0$ is a finite-rank operator on $L^2(\R^d)$, its density is defined as 
$$\rho_\gamma(x):=\gamma(x,x),\ \forall x\in\R^d,$$
where $\gamma(\cdot,\cdot)$ denotes the integral kernel of $\gamma$. We prove that for any $1/2<p<\ii$, there exists $C>0$ such that for any finite-rank $\gamma\ge0$ with $\rho_\gamma\in L^p(\R^d)$ we have
\begin{equation}\label{eq:LP-density}
  \frac1C\norm{\rho_\gamma}_{L^p(\R^d)} \le \norm{\sum_{j\in\Z}\rho_{P_j\gamma P_j}}_{L^p(\R^d)} \le C\norm{\rho_\gamma}_{L^p(\R^d)}.
\end{equation}
When $\gamma$ is a rank-one operator, this last inequality is equivalent to the usual Littlewood-Paley estimates \eqref{eq:LP}. Indeed, if $u$ with $\norm{u}_{L^2}=1$ belongs to the range of $\gamma$, then $\rho_\gamma=|u|^2$. 

The motivation to generalize the Littlewood-Paley decomposition to operator densities comes from many-body quantum mechanics. Indeed, a simple way to describe a system of $N$ fermions in $\R^d$ is via an orthogonal projection $\gamma$ on $L^2(\R^d)$ of rank $N$. The quantity $\rho_\gamma$ then describes the spatial density of the system. Variational or time-dependent models depending on $\gamma$ then typically include interactions between the particles via non-linear functionals of $\rho_\gamma$, like in Hartree-Fock models \cite{LieSim-77,BovPraFan-74,BovPraFan-76,Chadam-76}. As a consequence, $L^p$-properties of $\rho_\gamma$ are often needed to control these interactions. When $\gamma$ is a rank-one operator, these properties can be derived via Littlewood-Paley estimates (we typically think of Sobolev-type or Strichartz-type estimates). The estimate \eqref{eq:LP-density} allows to treat the rank $N$ case, and we illustrate this on the concrete example of the Lieb-Thirring inequality, which is a rank $N$ generalization of the Sobolev inequality.   

In Section \ref{sec:LP} we prove the inequality \eqref{eq:LP-density}. In Section \ref{sec:LT} we apply it to give a new proof of the Lieb-Thirring inequality.

\section{Littlewood-Paley for densities}\label{sec:LP}

In this section we prove the generalization of the Littlewood-Paley theorem to densities of operators. We will see that the proof is a simple adaptation of the proof of the usual Littlewood-Paley theorem. Thus, let us first recall briefly the proof of \eqref{eq:LP}. It is usually done via Khinchine's inequality \cite[Lemma 5.5]{MusSch-book}, see the proof of Theorem 8.3 in \cite{MusSch-book}: if one denotes by $(r_j)$ a sequence of independent random variables taking values in $\{\pm1\}$ and satisfying $\prob(r_j=\pm1)=1/2$, one has 
$$\frac1C\left(\sum_j |a_j|^2\right)^{p/2}\le \esp\left|\sum_j a_jr_j\right|^p\le C\left(\sum_j |a_j|^2\right)^{p/2},$$
for any set of coefficients $(a_j)\subset\C$, for some $C>0$, and for any $1\le p<\ii$. From this one deduces that 
$$\norm{\left(\sum_j |P_ju|^2\right)^{1/2}}_{L^p}^p\lesssim\esp\int_{\R^d}\left|\sum_j r_jP_ju(x)\right|^p\,dx.$$
The Fourier multiplier by the function $\xi\mapsto\sum_jr_j\psi_j(\xi)$ is bounded from $L^p(\R^d)$ to $L^p(\R^d)$ for any $1<p<\ii$, with a bound independent of the realization of the $(r_j)$. Indeed, one has to notice that for any given $\xi\in\R^d$, there are only a finite number of non-zero terms in the sum $\sum_jr_j\psi_j(\xi)$ (and this number only depends on $\psi$). The Mikhlin multiplier theorem \cite[Thm. 8.2]{MusSch-book} shows the boundedness of the Fourier multiplier. We deduce from all this the inequality
$$\esp\int_{\R^d}\left|\sum_j r_jP_ju(x)\right|^p\,dx\lesssim\esp\int_{\R^d}|u(x)|^p\,dx=\norm{u}_{L^p}^p.$$
The reverse inequality is done by a duality argument where the condition \eqref{eq:partition-unity} appears: we use the identity
$$\int_{\R^d}f\bar{g}\,dx=\sum_j\int_{\R^d}P_jf\bar{\tilde{P_j}g}\,dx,$$
where $\tilde{P_j}$ is another sequence of Littlewood-Paley multipliers such that $\tilde{P_j}P_j=P_j$ (which may be built from a $\tilde{\psi}$ which is identically 1 on the support of $\psi$). The fact that we cannot take $\tilde{P_j}=P_j$ is related to the deep fact that we cannot choose $P_j$ to be a projection (that is, we cannot take $\psi_j=\1(2^j\le\cdot<2^{j+1})$); indeed such a $P_j$ is not bounded on $L^p(\R^d)$ (except for $d=1$ or $p=2$) by Fefferman's famous result \cite{Fefferman-71}. 

The main result of this section is the following lemma.

\begin{lemma}\label{lem:LP-densities-sequences}
For any $1/2<p<\ii$, there exists $C>0$ such that for any $N\ge1$, for any $(\lambda_k)_{k=1}^N\subset\R_+$ and any functions $(u_k)_{k=1}^N$ in $L^{2p}(\R^d)$ we have 
\begin{equation}
 \frac1C\norm{\sum_k\lambda_k|u_k|^2}_{L^p}\le\norm{\sum_{j,k}\lambda_k|P_ju_k|^2}_{L^p}\le C\norm{\sum_k\lambda_k|u_k|^2}_{L^p}.
\end{equation}
\end{lemma}

Lemma \ref{lem:LP-densities-sequences} implies the Littlewood-Paley decomposition \eqref{eq:LP-density} for densities using the spectral decomposition of $\gamma$. We first need a version of Khinchine's inequality for tensor products, which is proved for instance in \cite[Appendix D]{Stein-book}. We however include a proof here for completeness.

\begin{lemma}\label{lemma:khinchine-tensor}
 Let $(a_{j,k})\subset\C$ a sequence of coefficients and $(r_j)$ a sequence of independent random variables such that $\prob(r_j=\pm1)=1/2$. Then, we have 
 $$\left(\sum_{j,k}|a_{j,k}|^2\right)^{p/2}\lesssim \esp\left|\sum_{j,k}a_{j,k}r_jr_k\right|^p,$$
 for all $1\le p<\ii$, where the implicit constant is independent of $(a_{j,k})$.
\end{lemma}

\begin{remark}
 The reverse inequality also holds; we however do not need it here.
\end{remark}

\begin{remark}
 This inequality does not follow from the Khinchine inequality from abstract arguments because the sequence $(r_jr_k)$ is not independent anymore: knowing $r_1r_2$ and $r_1r_3$ implies that we know $r_2r_3$ as well. 
\end{remark}

\begin{proof}[Proof of Lemma \ref{lemma:khinchine-tensor}]
  We only prove it for $1\le p\le 2$, which is sufficient since $\esp|g|^p\ge(\esp|g|^2)^{p/2}$ for $p\ge2$. We first apply Khinchine's inequality with respect to the random parameter associated to $(r_k)$:
  $$\esp\left|\sum_{j,k}a_{j,k}r_jr_k\right|^p\gtrsim\esp_1\left(\sum_k\left|\sum_ja_{j,k}r_j\right|^2\right)^{p/2},$$
  where $\esp_1$ denotes the expectation with respect to the random parameter associated to $(r_j)$. Since $p/2\le1$, we may apply the reverse Minkowski inequality\footnote{Stating that $\norm{\sum_k f_k}_{L^{p/2}}\ge\sum_k\norm{f_k}_{L^{p/2}}$ for any $f_k\ge0$.} to infer that
  $$\esp_1\left(\sum_k\left|\sum_ja_{j,k}r_j\right|^2\right)^{p/2}\ge\left(\sum_k\left(\esp_1\left|\sum_ja_{j,k}r_j\right|^p\right)^{2/p}\right)^{p/2}.$$
  Using a second time Khinchine's inequality leads to 
  $$\left(\sum_k\left(\esp_1\left|\sum_ja_{j,k}r_j\right|^p\right)^{2/p}\right)^{p/2}\gtrsim\left(\sum_{j,k}|a_{jk}|^2\right)^{p/2}.$$
\end{proof}

From this tensorized Khinchine inequality, we deduce one side of the desired inequality.

\begin{lemma}\label{lemma:weak-ineq}
 Let $(\lambda_k)\subset\R_+$ a finite sequence of coefficients and $(u_k)$ a finite sequence in $L^{2p}(\R^d)$. Then, we have 
 \begin{equation}
  \norm{\sum_{j,k}\lambda_k|P_ju_k|^2}_{L^p}\lesssim\norm{\sum_k\lambda_k|u_k|^2}_{L^p},
 \end{equation}
 for all $1/2<p<\ii$, where the implicit constant is independent of $(\lambda_k)$, $(u_k)$.
\end{lemma}

\begin{proof}
 By Lemma \ref{lemma:khinchine-tensor},
 $$\norm{\sum_{j,k}\lambda_k|P_ju_k|^2}_{L^p}^p\lesssim\esp\int_{\R^d}\left|\sum_{j,k}\lambda_k^{1/2}r_jr_kP_ju_k(x)\right|^{2p}\,dx.$$
 By the boundedness of the Fourier multiplier by $\xi\mapsto\sum_j r_j\psi_j(\xi)$ on $L^{2p}$, we have 
 $$\esp\int_{\R^d}\left|\sum_{j,k}\lambda_k^{1/2}r_jr_kP_ju_k(x)\right|^{2p}\,dx\lesssim\esp\int_{\R^d}\left|\sum_k\lambda_k^{1/2}r_ku_k(x)\right|^{2p}\,dx.$$
 Applying again Khinchine's inequality, we have
 $$\int_{\R^d}\esp\left|\sum_k\lambda_k^{1/2}r_ku_k(x)\right|^{2p}\,dx\lesssim\int_{\R^d}\left(\sum_k\lambda_k|u_k(x)|^2\right)^p\,dx.$$
\end{proof}

The other side of the inequality uses Lemma \ref{lemma:weak-ineq}.

\begin{lemma}
 Let $(\lambda_k)\subset\R_+$ a finite sequence of coefficients and $(u_k)$ a finite 
 set of functions in $L^{2p}(\R^d)$. Then, we have
 \begin{equation}\label{eq:LP-sequences-2nd}
    \norm{\sum_k\lambda_k|u_k|^2}_{L^p}\lesssim\norm{\sum_{j,k}\lambda_k|P_ju_k|^2}_{L^p},
 \end{equation}
  for all $1/2<p<\ii$, where the implicit constant is independent of $(\lambda_k)$, $(u_k)$. 
\end{lemma}

\begin{remark}
The right side of \eqref{eq:LP-sequences-2nd} is well-defined due to Lemma \ref{lemma:weak-ineq}.
\end{remark}

\begin{proof}
  For any $V\ge0$, we have
  \begin{align*}
    \int_{\R^d}\left(\sum_k\lambda_k|u_k(x)|^2\right) V(x)\,dx &=\sum_k\lambda_k\int_{\R^d}\bar{u_k(x)}V(x)u_k(x)\,dx \\
    &= \sum_{j,k}\lambda_k\int_{\R^d}\bar{P_ju_k(x)}\tilde{P_j}Vu_k(x)\,dx,
  \end{align*}
  where the sequence $\tilde{P_j}$ was defined earlier. By Hölder's inequality,
  \begin{align*}
   \sum_{j,k}\lambda_k\int_{\R^d}\bar{P_ju_k(x)}\tilde{P_j}Vu_k(x)\,dx &\le \int_{\R^d}\left(\sum_{j,k}\lambda_k|P_ju_k(x)|^2\right)^{1/2}\left(\sum_{j,k}\lambda_k|\tilde{P_j}Vu_k(x)|^2\right)^{1/2}\\   
  &\le\norm{\sum_{j,k}\lambda_k|P_ju_k(x)|^2}_{L^p}^{1/2}\norm{\sum_{j,k}\lambda_k|\tilde{P_j}Vu_k(x)|^2}_{L^{p/(2p-1)}}^{1/2}.
  \end{align*}
 By Lemma \ref{lemma:weak-ineq}, using that $p/(2p-1)>1/2$, we have
 $$\norm{\sum_{j,k}\lambda_k|\tilde{P_j}Vu_k(x)|^2}_{L^{p/(2p-1)}}\lesssim\norm{V^2\sum_k\lambda_k|u_k|^2}_{L^{p/(2p-1)}},$$
 which leads to the desired result by choosing $V=(\sum_k\lambda_k|u_k|^2)^{p-1}$. 
\end{proof}

\section{Application: Lieb-Thirring inequalities}\label{sec:LT}

In this section, we explain how to use the Littlewood-Paley decomposition \eqref{eq:LP-density} to provide a simple proof of the Lieb-Thirring inequality. We first compare the Littlewood-Paley decompositions \eqref{eq:LP} and \eqref{eq:LP-density}, and argue why they cannot be used in the same way.

\subsection{Comparison of the two Littlewood-Paley decompositions}\label{sec:comparison}

The Lieb-Thirring inequality generalizes to densities of operators the Gagliardo-Nirenberg-Sobolev inequality
\begin{equation}
 \norm{u}_{L^{2+4/d}}\lesssim\norm{u}_{L^2}^{\frac{2}{d+2}}\norm{\nabla u}_{L^2}^{\frac{d}{d+2}},\ \forall u\in H^1(\R^d).
\end{equation}
This last inequality can be proved very easily using the usual Littlewood-Paley decomposition \eqref{eq:LP}. Indeed, by Hölder's inequality we have
\begin{align*}
  \norm{P_ju}_{L^{2+4/d}} &\le\norm{P_ju}_{L^2}^{\frac{d}{d+2}}\norm{P_ju}_{L^\ii}^{\frac{2}{d+2}}\\
  &\lesssim\norm{P_ju}_{L^2}^{\frac{d}{d+2}}\norm{\cF(P_ju)}_{L^1}^{\frac{2}{d+2}}\\
  &\lesssim 2^{\frac{d}{d+2}j}\norm{P_ju}_{L^2}\\
  &\lesssim\norm{P_ju}_{L^2}^{\frac{2}{d+2}}\norm{\nabla P_j u}_{L^2}^{\frac{d}{d+2}},
\end{align*}
meaning that the Gagliardo-Nirenberg-Sobolev inequality is immediate for frequency-localized functions. To get it for any function, we use the Littlewood-Paley decomposition \eqref{eq:LP} and obtain
\begin{align*}
 \norm{u}_{L^{2+4/d}}^2 &\lesssim\sum_j\norm{P_ju}_{L^{2+4/d}}^2\\
  &\lesssim \sum_j\norm{P_ju}_{L^2}^{\frac{4}{d+2}}\norm{\nabla P_j u}_{L^2}^{\frac{2d}{d+2}}\\
  &\lesssim \left(\sum_j\norm{P_ju}_{L^2}^2\right)^{\frac{2}{d+2}}\left(\sum_j\norm{\nabla P_j u}_{L^2}^2\right)^{\frac{d}{d+2}}\\
  &\lesssim \norm{u}_{L^2}^{\frac{4}{d+2}}\norm{\nabla u}_{L^2}^{\frac{2d}{d+2}}.
\end{align*}
We see here the power of the Littlewood-Paley decomposition: it allows to deduce functional inequalities from their version for frequency-localized functions. This has been used in several contexts, for instance concerning Strichartz inequalities \cite{KeeTao-98,BurGerTzv-04}. In particular, notice that we have used something much weaker than the Littlewood-Paley decomposition, namely the inequality
$$\norm{u}_{L^p}^2\lesssim\sum_j\norm{P_ju}_{L^p}^2,$$
which follows from \eqref{eq:LP} by a triangle inequality. We now explain why the same strategy does not work in the context of the Lieb-Thirring inequality. This inequality reads
$$\tr(-\Delta)\gamma\gtrsim\int_{\R^d}\rho_\gamma(x)^{1+\frac2d}\,dx,$$
for any finite-rank $0\le\gamma\le1$. To see that it is indeed a generalization of the Gagliardo-Nirenberg-Sobolev inequality, notice that it is equivalent to the inequality
$$\int_{\R^d}\left(\sum_{k=1}^N\lambda_k|u_k(x)|^2\right)^{1+\frac2d}\,dx\lesssim\sum_{k=1}^N\lambda_k\int_{\R^d}|\nabla u_k(x)|^2\,dx,$$
for any $(\lambda_k)\subset\R_+$, $(u_k)\subset H^1(\R^d)$ orthonormal in $L^2(\R^d)$, and any $N\ge1$. The usual Gagliardo-Nirenberg-Sobolev inequality thus corresponds to the particular case $N=1$ of the Lieb-Thirring inequality. However, the Lieb-Thirring inequality does not follow from the Gagliardo-Nirenberg-Sobolev and the triangle inequalities, they only imply that 
$$\int_{\R^d}\left(\sum_{k=1}^N\lambda_k|u_k(x)|^2\right)^{1+\frac2d}\,dx\lesssim\left(\sum_{k=1}^N\lambda_k\right)^{\frac2d}\sum_{k=1}^N\lambda_k\int_{\R^d}|\nabla u_k(x)|^2\,dx,$$
which is weaker than the Lieb-Thirring inequality, especially for large $N$. Let us notice that Frank, Lieb, and Seiringer have proved in \cite{FraLieSei-10} an equivalence between the Gagliardo-Nirenberg-Sobolev and (the dual version of) the Lieb-Thirring inequality.

Again, for frequency-localized $\gamma$, this inequality is elementary: the constraint $0\le\gamma\le1$ implies that $0\le P_j\gamma P_j\le P_j^2$ and hence $0\le\rho_{P_j\gamma P_j}(x)\lesssim 2^{dj}$ for all $x\in\R^d$. As a consequence,
$$\norm{\rho_{P_j\gamma P_j}}_{L^{1+2/d}}\le\norm{\rho_{P_j\gamma P_j}}_{L^1}^\frac{d}{d+2}\norm{\rho_{P_j\gamma P_j}}_{L^\ii}^\frac{2}{d+2}\lesssim 2^{\frac{2d}{d+2}j}(\tr P_j\gamma P_j)^\frac{d}{d+2}\lesssim(\tr(-\Delta)P_j\gamma P_j)^{\frac{d}{d+2}},$$
which is exactly the Lieb-Thirring inequality. Here, we used the fact that $\int\rho_\gamma=\tr\gamma$. Using the same idea as in the proof of the Gagliardo-Nirenberg-Sobolev inequality, we find that for any $\gamma$,
$$\norm{\rho_\gamma}_{L^{1+2/d}}\lesssim\sum_j\norm{\rho_{P_j\gamma P_j}}_{L^{1+2/d}}\lesssim\sum_j(\tr(-\Delta)P_j\gamma P_j)^{\frac{d}{d+2}},$$
which we cannot sum. Indeed, the inequality
$$\sum_j(\tr(-\Delta)P_j\gamma P_j)^{\frac{d}{d+2}}\le\left(\sum_j\tr(-\Delta)P_j\gamma P_j\right)^{\frac{d}{d+2}}\sim(\tr(-\Delta)\gamma)^{\frac{d}{d+2}}$$
is of course wrong because $d/(d+2)<1$. We thus see the difference between the applications of the Littlewood-Paley decompositions for functions or for densities of operators: one cannot directly resum the frequency-localized inequalities in the context of operators. Of course, the reason behind it is the use of the rough triangle inequality $\norm{\rho_\gamma}_{L^p}\lesssim\sum_j\norm{\rho_{P_j\gamma P_j}}_{L^p}$, which one should not do for operators. We now explain how to go beyond this difficulty.

\subsection{Proof of the Lieb-Thirring inequality}

Let us prove the Lieb-Thirring inequality using the Littlewood-Paley decomposition for densities. Hence, let $0\le\gamma\le1$ an operator on $L^2(\R^d)$, which we may assume to be of finite rank. Since $1=\sum_j P_j$ with $P_j\ge0$, we deduce that $1\ge\sum_j P_j^2$. We thus have
\begin{equation}\label{eq:LT-first-ineq}
  \tr(-\Delta)\gamma\ge\sum_j\tr\sqrt{\gamma}P_j(-\Delta)P_j\sqrt{\gamma}\gtrsim \sum_j2^{2j}\tr\sqrt{\gamma} P_j^2\sqrt{\gamma}=\int_{\R^d}\sum_j 2^{2j}\rho_{P_j\gamma P_j}(x)\,dx.
\end{equation}

\begin{lemma}\label{lemma:sequences}
 Let $(\alpha_j)_{j\in\Z}$ a sequence of real numbers satisfying $0\le\alpha_j\le 2^{jd}$ for all $j$. Then, we have the inequality
 $$\left(\sum_j\alpha_j\right)^{1+\frac2d}\lesssim\sum_j2^{2j}\alpha_j.$$
\end{lemma}

Let us first notice that the lemma implies the Lieb-Thirring inequality: indeed, since $0\le\gamma\le1$ we deduce that $0\le P_j\gamma P_j\le P_j^2$ and hence $0\le\rho_{P_j\gamma P_j}(x)\lesssim 2^{jd}$ for all $x\in\R^d$. Hence, from the Lemma and \eqref{eq:LT-first-ineq} we deduce that 
$$\tr(-\Delta)\gamma\gtrsim\int_{\R^d}\left(\sum_j\rho_{P_j\gamma P_j}(x)\right)^{1+\frac2d}dx\gtrsim\int_{\R^d}\rho_\gamma(x)^{1+\frac2d}dx,$$
where in the last inequality we used the Littlewood-Paley theorem for densities. Let us now prove the lemma.

\begin{proof}[Proof of Lemma \ref{lemma:sequences}]
 We split the following sum as
 $$\sum_j\alpha_j=\sum_{j\le J}\alpha_j + \sum_{j>J}\alpha_j.$$
 We estimate the first sum using that $0\le\alpha_j\le2^{jd}$:
 $$\sum_{j\le J}\alpha_j\lesssim 2^{dJ},$$
 and the second sum is estimated in the following way:
 $$\sum_{j>J}\alpha_j\le 2^{-2J}\sum_j2^{2j}\alpha_j.$$
 We thus find that for all $J$,
 $$\sum_j\alpha_j\lesssim 2^{dJ}+2^{-2J}\sum_j 2^{2j}\alpha_j.$$
 Optimizing over $J$ leads to the result.
\end{proof}

Of course, the same strategy of proof allows to obtain more general inequalities of the type
$$\tr(-\Delta)^b\gamma\gtrsim\int_{\R^d}\rho_\gamma(x)^{1+\frac{2b}{d+2a}}\,dx,$$
for all $0\le\gamma\le(-\Delta)^a$, with $b\ge0$ and $a>-d/2$. In particular, the case $d\ge3$, $a=-1$, $b=1$ is due to Rumin \cite{Rumin-10} and was shown to be equivalent to the CLR inequality by Frank \cite{Frank-14}. Our method is similar to the one used by Rumin, except that he uses a continuous decomposition
$$-\Delta=\int_0^\ii\1(-\Delta>\tau)\,d\tau$$
instead of a dyadic decomposition coming from Littlewood-Paley. Rumin's method is actually far more powerful when dealing with these kind of inequalities, and was shown to work when replacing $-\Delta$ by general $a(-i\nabla)$ by Frank \cite{Frank-14}. The dyadic decomposition seems useless in these more general cases since it does not distinguish the high/low values of $a$. We expect that the Littlewood-Paley decomposition might be useful when one wants to exploit the ``almost orthogonality'' between the blocks $(P_j)$: we have $P_jP_k=0$, except for finite number of blocks, a phenomenon which does not appear in Rumin's decomposition. This orthogonality might be useful when dealing with higher Schatten spaces $\gS^\alpha$ compared to the trace-class $\gS^1$ which appears for instance in the Lieb-Thirring inequality. We hope to find such applications in the future. 

\medskip

\noindent\textbf{Acknowledgments.} The author is grateful to Rupert Frank for useful discussions. Financial support from the ERC MNIQS-258023 is acknowledged. 

%

\begin{thebibliography}{10}

\bibitem{BovPraFan-74}
{\sc A.~Bove, G.~{Da Prato}, and G.~Fano}, {\em An existence proof for the
  {H}artree-{F}ock time-dependent problem with bounded two-body interaction},
  Commun. Math. Phys., 37 (1974), pp.~183--191.

\bibitem{BovPraFan-76}
\leavevmode\vrule height 2pt depth -1.6pt width 23pt, {\em On the
  {H}artree-{F}ock time-dependent problem}, Commun. Math. Phys., 49 (1976),
  pp.~25--33.

\bibitem{BurGerTzv-04}
{\sc N.~Burq, P.~G{\'e}rard, and N.~Tzvetkov}, {\em Strichartz inequalities and
  the nonlinear {S}chr\"odinger equation on compact manifolds}, Amer. J. Math.,
  126 (2004), pp.~569--605.

\bibitem{Chadam-76}
{\sc J.~Chadam}, {\em The time-dependent {H}artree-{F}ock equations with
  {C}oulomb two-body interaction}, Commun. Math. Phys., 46 (1976), pp.~99--104.

\bibitem{Fefferman-71}
{\sc C.~Fefferman}, {\em The multiplier problem for the ball}, Ann. of Math.
  (2), 94 (1971), pp.~330--336.

\bibitem{Frank-14}
{\sc R.~L. Frank}, {\em Cwikel's theorem and the {CLR} inequality}, J. Spectr.
  Theory, 4 (2014), pp.~1--21.

\bibitem{FraLieSei-10}
{\sc R.~L. Frank, E.~H. Lieb, and R.~Seiringer}, {\em Equivalence of {S}obolev
  inequalities and {L}ieb-{T}hirring inequalities}, in X{VI}th {I}nternational
  {C}ongress on {M}athematical {P}hysics, World Sci. Publ., Hackensack, NJ,
  2010, pp.~523--535.

\bibitem{FraJawWei-book}
{\sc M.~Frazier, B.~Jawerth, and G.~Weiss}, {\em Littlewood-{P}aley theory and
  the study of function spaces}, vol.~79 of CBMS Regional Conference Series in
  Mathematics, Published for the Conference Board of the Mathematical Sciences,
  Washington, DC; by the American Mathematical Society, Providence, RI, 1991.

\bibitem{KeeTao-98}
{\sc M.~Keel and T.~Tao}, {\em Endpoint {S}trichartz estimates}, Amer. J.
  Math., 120 (1998), pp.~955--980.

\bibitem{LieSim-77}
{\sc E.~H. Lieb and B.~Simon}, {\em The {H}artree-{F}ock theory for {C}oulomb
  systems}, Commun. Math. Phys., 53 (1977), pp.~185--194.

\bibitem{MusSch-book}
{\sc C.~Muscalu and W.~Schlag}, {\em Classical and multilinear harmonic
  analysis. {V}ol. {I}}, vol.~137 of Cambridge Studies in Advanced Mathematics,
  Cambridge University Press, Cambridge, 2013.

\bibitem{Rumin-10}
{\sc M.~Rumin}, {\em Spectral density and {S}obolev inequalities for pure and
  mixed states}, Geom. Funct. Anal., 20 (2010), pp.~817--844.

\bibitem{Stein-book}
{\sc E.~Stein}, {\em Singular Integrals and Differentiability Properties of
  Functions}, Princeton Mathematical Series, 30, Princeton University Press,
  1970.

\end{thebibliography}

\end{document}